\documentclass[12pt]{article}
\usepackage{amsmath}
\usepackage{amssymb}
\usepackage{amsthm}
\usepackage{amsfonts}
\usepackage{graphicx}
\usepackage{dsfont} 

\usepackage{bbm,epsfig,graphics,epic,color,rotating,color}
\numberwithin{equation}{section}

\newcommand{\R}{{\mathbb R}}
\newcommand{\Z}{{\mathbb Z}}

\renewcommand{\S}{{\mathbb S}}

\newcommand{\F}{{\mathcal{F}}}
\newcommand{\N}{{\mathcal{N}}}

\newcommand{\ed}{\mathrm{d}}

\newcommand{\reff}[1]{(\ref{#1})}
\newcommand{\pch}[1]{{\color{blue}#1}}

\newcommand{\E}{{\mathsf{E}}}
\renewcommand{\k}{\varkappa}
\newcommand\e{\varepsilon}
\newcommand\la{\lambda}
\renewcommand\phi{\varphi}

\newcommand{\be}{\begin{equation}}
\newcommand{\ee}{\end{equation}}
\newcommand{\bel}[1]{\begin{equation}\label{#1}}
\newcommand{\bea}{\begin{eqnarray}}
\newcommand{\eea}{\end{eqnarray}}
\newcommand{\balign}{\begin{align}}
\newcommand{\ealign}{\end{align}}
\newcommand{\ba}{\begin{array}}
\newcommand{\ea}{\end{array}}
\newcommand{\bfig}{\begin{figure}}
\newcommand{\efig}{\end{figure}}
\newcommand{\bra}[1]{\mbox{$\langle \, {#1}\, |$}}
\newcommand{\ket}[1]{\mbox{$| \, {#1}\, \rangle$}}
\newcommand{\dbra}[1]{\mbox{$\langle\langle \, {#1}\, |$}}
\newcommand{\dket}[1]{\mbox{$| \, {#1}\, \rangle\rangle$}}

\newcommand{\inprod}[2]{\mbox{$\langle \, {#1} \, | \, {#2} \, \rangle$}}
\newcommand{\dinprod}[2]{\mbox{$\langle\langle  \, {#1} \, | \, {#2} \, \rangle\rangle$}}

\newcommand{\C}{{\mathbb C}}
\newcommand{\ints}{{\mathbb N}}
\newcommand{\rme}{\mathrm{e}}
\newcommand{\rmi}{\mathrm{i}}

\newcommand{\eref}[1]{(\ref{#1})}

\newtheorem{theorem}{Theorem}[section]

\newtheorem{remark}[theorem]{Remark}

\theoremstyle{definition}

\title{Large Fluctuations of Radiation in Stochastically Activated Two-Level Systems}

\author{ E. Pechersky$^{1}$,\and S. Pirogov$^{1}$,\and G. M. Sch\"utz$^{2,3}$, \and A.~Vladimirov$^{1}$  and  A. Yambartsev$^4$ }
\date{}

\begin{document}

\maketitle

$^1$ Institute for Information Transmission Problems, 19, Bolshoj Karetny, Moscow, 127994, RF

$^2$Institute of Complex Systems II, Forschungszentrum J\"ulich, 52425 J\"ulich, Germany Email: g.schuetz@fz-juelich.de

$^3$Interdisziplin\"ares Zentrum f\"ur Komplexe Systeme, Universit\"at Bonn, Br\"uhler Str. 7, 53119 Bonn, Germany

$^4$IME, University of Sao Paulo (USP), Sao Paulo 05508-090,
SP, Brazil

\begin{abstract} We study large fluctuations of emitted radiation in the system of $N$ non-interacting two-level atoms. Two methods are used to calculate the probability of large fluctuations and the time dependence of excitation and emission. The first method is based on the large deviation principle for Markov processes. The second one uses an analogue of the quantum formalism for classical probability problems.

Particularly we prove that in a large fluctuation limit approximately half of the atoms are excited. This fact is independent on the fraction of the excited atoms in equilibrium.

\end{abstract}



\section{Motivations}\label{Intro}

In recent years there has been considerable interest in studying optimal realizations of large deviations in stochastic interacting
particle systems, see e.g. \cite{Beli13a,Bodi05,Bodi06} for optimal profiles in current large deviations of the
asymmetric simple exclusion process
or \cite{Harr06,Harr13} for condensation in the zero-range process. In addition to the mathematical
interest in understanding how large deviations are typically realized (when they occur), motivation comes also from the fact that
such large deviations are studied experimentally in order to test the validity of fluctuation theorems such as the Gallavotti-Cohen
theorem \cite{Gall95b}, the Lebowitz-Spohn results \cite{Lebo99} or the Jarzynski relation \cite{Jarz97}. Generally, fluctuation theorems predict the ratio of probabilities
of a positive large deviation of some
observable to a negative large deviation of the same absolute magnitude \cite{Chet11,Lebo99,Schu14}.
Particularly simple and therefore transparent systems where such relations can be probed experimentally
are two-level systems where each microscopic
entity can be either in a ground state or in an excited state. As an example we mention the optical excitation of an ensemble of
single-defect centres in diamond.  Measurements of the dissipated work \cite{Schuler05} and entropy \cite{Tietz06b} agree
with the expected integral fluctuation relations. Fluctuation relations for driven two-level systems have been investigated recently also
in the context of biological motors \cite{Chvo10} and quantum two-level systems \cite{Liu13}.

Here we  address the question of optimal realizations
of a large deviation in an ensemble of two-level systems. We consider a scenario where every two-level system in the ensemble of (independent)
two-level systems can be excited by an external source of energy and later relaxes to its ground state by the emission of radiation.
Conditioning on a large deviation of the emitted radiation in a time interval $[0,T]$ we ask how the number of excited states $\Xi_N(t)$
and also how the accumulative radiation $\Upsilon_N(t)$ will typically behave during this time interval assuming that the number $N$ of elements of the two-level system is large.
In particular, we wish to answer the question whether such
a large deviation is typically realized by the typical behaviour for most of the time and only some strong bursts of radiation, or by
non-typical behaviour throughout $[0,T]$, or by both, i.e., by non-typical behaviour together with  bursts.

For definiteness we refer to microscopic two-level systems as atoms.
To keep the model simple we study the case where both excitation and emission occur after exponentially distributed random times
with parameters $\lambda$ and $\mu$ respectively. It will transpire that the behaviour of the ensemble of two-level systems is non-trivial: the optimal realization of a
large deviation involves
non-typical radiation intensity throughout $[0,T]$. If $T$ is large then at intermediate times the radiation activity is approximately constant (and different
from its typical value), and one has moderate bursts or suppression at the beginning or the end (or both) of the
time-interval $[0,T]$, depending on the initial state and on the magnitude of the large deviation that is under consideration.
The ``moderate bursts'' are defined as an expected increase in the radiation intensity $\dot{\Upsilon}_N(t)$ during a short (exponentially bounded)
interval of time.
The probabilistic interpretation of this behaviour is that in order to generate a large deviation from typical
behaviour in a long time interval $T$ it is less costly to have a relatively small, but lasting, deviation from non-typical behaviour
than to have almost permanent typical behaviour, punctuated by occasional strong deviations (strong bursts of radiation). The (moderate)
bursts at initial and/or final times realize the transition from the initial state (and/or to the final state) to the intermediate mildly nontypical
behaviour.

We address this problem both from a macroscopic Hamiltonian perspective  and from a microscopic approach employing
generating function techniques. The latter indicates that qualitatively similar behaviour is expected for any number of excited
states per atom.

\section{Macroscopic approach. Large Deviations.}

In this section fluctuations of an ensemble of two-level atoms are described by means of the large deviation principle.
The goal is to study a rare event where a very large emission is happening during time interval $[0,T]$. More exactly, we find a mean dynamics of the number of emissions and the number of excited atoms during $[0,T]$ conditioned on the rare event mentioned above.

Consider an ensemble (a sample) of $N$ atoms where each atom can be either in a
ground state or in an excited state. The sample is under a radiation by
steady flux of photons. Each atom in the ground state can absorb a
photon transiting to the exited state. An atom in the excited state
can emit a photon and consequently undergoes a transition to the ground
state. There are no interactions between the atoms. Therefore all atom
transformations are independent of each other. Given the
parameters of mean times that an atom spends in the ground and the excited
states it is easy to find the mean dynamics of the number of excited atoms
and of the cumulative number of the photon emissions on $[0,T]$.

We show that the process path given the large emission on $[0,T]$ is close to a linear function. Moreover, under the large emission the conditional input and output flows are close to each other in the interior of $[0,T]$. This fact holds for any stationary regime of both flows.

We describe this dynamics by a two-dimensional Markov process.
For any $t\in[0,T]$, the components of the process describe the number
of the excited atoms at moment $t$ and the cumulative number of emission
events during $[0,t]$. In terms of the Markov process, our goal is to
describe the shape of paths of the Markov process given the large cumulative
number of emission events on $[0,T]$.

A tool we used for the analysis is the large deviation theory, which studies the probabilities of rare events, and also describes how these rare events evolve. The main tool of the large deviation theory is a theorem called the large deviation principle. There exists many versions of the large deviation principle for different stochastic models.
Here we apply the large deviation principle to Markov processes following \cite{FeK}.

\subsection{The model}\label{def} We formalize our model as follows. Consider $N$ independent random variables (atoms in what follows) that take two values $0$ and $1$. We say that the atom is in the ground state if the corresponding random variable is $0$, otherwise the atom is excited.  We study time dynamics of every atom. As time passes, any atom randomly changes its state. The stochastic dynamics of a single atom is a Markov process on the state space $\{0,1\}$. The transition of the atom from the  excited state to the ground state $1\to 0$ is called an \textit{emission}. The opposite transition is called an \textit{excitation}.

A stochastic dynamics of the ensemble $\Sigma= (\sigma_1,...,\sigma_N)$ of $N$ two-level atoms where $\sigma_i\in\{0,1\}$ is described by a two dimensional Markov process $\Gamma_N(t)=\big(\Xi_N(t),\Upsilon_N(t)\big)$, where  $\Xi_N(t)$ is the number of excited atoms at time $t$ and $\Upsilon_N(t)$ is the cumulative number of emitted photons on the interval $[0,t]$. The components $\Xi_N(t)$ and $\Upsilon_N(t)$ are strongly correlated.

\subsubsection{Excitation component}\label{exi}
The process $\Xi_N(t)$ takes its values in the set $\mathcal{N}=\{0,1,2,...,N\}$. A random event $(\Xi_N(t)=n),\ n\in\mathcal{N}$ means that the number of excited atoms is equal to $n$. Then $N-n$ is the number of the atoms in their ground states.  The process $\Xi_N(t)$ is a jump Markov process. Every jump is either $+1$ or $-1$, either enlarging or reducing the number of the excited atoms by one.

An operator semigroup $(\mathbf{P}_t,\ t\in\R_+)$ drives the dynamics of
$\Xi_N(t)$ with the initial value  $\Xi_N(0)=x\in\N$. Namely,
\begin{equation}
\mathbf P_tG(x)=\E _xG(x_t),
\end{equation}
where $x_t$ is the trajectory of Markov process starting from $x_0=x$, and where $G\in \F$, i.e.
the operators $\mathbf{P}_t$ act on the space $\F$ of functions
$\N\to\R$. The semigroup can be represented by an infinitesimal operator
$\mathbf L$ as
\begin{equation}\label{1o0}
\mathbf{P}_tG=e^{t\mathbf{L}}G
\end{equation}
for $G\in\F$. The infinitesimal operator $\mathbf{L}$ of $\Xi_N(t)$ is
determined by two parameters $\lambda>0$  and $\mu>0$ as follows:
\begin{equation}
\mathbf{L}G(n)=
\lambda(N- n)\left[G(n+1)- G(n) \right]+\mu n\left[ G(n-1)- G(n)\right].
\end{equation}
The constants $\la$ and $\mu$ are intensities of the jumps $+1$ and $-1$
respectively.

The process $\Xi_N(t)$ is ergodic and reversible with the stationary
distribution
\begin{equation}\label{2o1}
\pi(n):=\Pr(\Xi_N(t)=n)=\frac{\binom{N}{n}\gamma^n}{(1+\gamma)^N},
\end{equation}
where $\gamma=\lambda/\mu$. The distribution  $\pi$ can be found from
the detailed balance equations
\begin{equation}\label{2o2}
\lambda (N-n)\pi(n)=\mu(n+1)\pi(n+1),  \mbox{ for } n<N.
\end{equation}

\subsubsection{Emission component}
The second component $\Upsilon_N(t)$ describes the number of emission
events and is equal to the number of negative jumps of the process $\Xi_N$
on the time interval $[0,t]$. Namely, let $k(t)=k_+(t)+k_-(t)$ be the
total number of jumps of $\Xi_N$ up to the moment $t\in[0,T]$, where
$k_+(t)$ and $k_-(t)$ are the numbers of positive and negative jumps,
respectively. Then we define
$$
\Upsilon_N(t)=-k_-(t).
$$
It is assumed that $k_+(0)=k_-(0)=0$. Note that
$\Xi_N(t)=\Xi_N(0)+k_+(t)-k_-(t)$.

\subsubsection{The process}
The two-dimensional process
\begin{equation}\label{3o0}
\Gamma_N(t)=(\Xi_N(t),\Upsilon_N(t))
\end{equation}
is a Markov process with dependent components taking their values in
$\N\times \Z_-$, where $\Z_-=\{...-2,-1,0\}$. The process $\Gamma_N(t)$ is a
non-equilibrium jump process. Observe that the process has only two types of
jumps $\{(1,0), (-1,-1)\}$.  The infinitesimal operator of $\Gamma_N(t)$ is
\begin{equation}\label{3o1}
\mathbf{L}_{\Gamma}G(n,m) =\lambda(N-n)[G(n+1,m)-G(n,m)] +\mu n[G(n-1,m-1)-G(n,m)],
\end{equation}
where $ n\in\N, m\in\Z_- $, and $\lambda, \mu$ are positive parameters
introduced in Section~\ref{exi}, $G(n,m)$ belongs to the function
space $\F_\Gamma$ of functions $\N\times\Z_-\to \R$.

\subsection{Large deviations}
The goal of our studies is to find the probability of a very large
emission and a large deviation path of the process which produces very
large emission during the interval $[0,T]$. It is assumed that the number
$N$ of the atoms is large and grows to infinity. The theory of large deviations 
allows one to solve both mentioned tasks: to find the asymptotics
of the large emission probability as $N\to\infty$ as well as the path
of the dynamics that realizes the large deviation on the interval
$[0,T]$. Solving these problems by the method of large deviations
we follow constructions and results in \cite{FeK} (see also Appendix in
\cite{Enter}).

In order to study the asymptotics in $N$ it is convenient to consider
the scaled process \bel{scaleproc}
\gamma_{N}(t)=\left(\xi_{N}(t)=\frac{1}{N}\Xi_N(t),\zeta_{N}(t)=\frac{1}{N}\Upsilon_N(t)\right)
\ee instead of \reff{3o0}. The scaled process $\gamma_N=(\xi_N,\zeta_N)$
takes  values in $\mathbf D_{N}=\left( \frac{1}{N}\N
\times\frac{1}{N}\Z_-\right) $. The process $\gamma_N$ is a jump process with
two types of jumps: $\left(\frac1N,0\right)$ and
$\left(-\frac1N,-\frac1N\right)$ with intensities $\lambda$ and $\mu$
respectively. Let $(x,y)\in\mathbf D_{N}$ and $\alpha=\frac 1N$. Then the
infinitesimal operator of $\gamma_N$ is
\begin{eqnarray*}
\mathbf{L}_{\gamma_N}G\left(x_N,y_N\right) &=& N\lambda(1-x_N)[G(x_N+\alpha,y_N)-G(x_N,y_N)] \\
&+& N\mu x_N[G(x_N-\alpha,y_N-\alpha)-G(x_N,y_N)].
\end{eqnarray*}

Since $\mathbf D_{N}\subset\mathbf D=[0,1] \times\R_-$ for every $ N $, we
will assume that the processes $ \gamma_{N} $ take their values in $
\mathbf D $. Let  $\mathcal F_2$ be the space of smooth functions
$\mathbf D\to \R$. Assume that the sequence $(x_N, y_N) \in {\mathbf
D}_N$ converges to some $(x,y)\in \mathbf D$, as $N\to\infty$. Then
\begin{equation}\label{LimOp}
\lim_{N\to\infty} \mathbf{L}_{\gamma_N}G(x_N,y_N)=\la(1-x)\frac{\partial }{\partial x}G(x,y)-\mu x\left(\frac{\partial}{\partial x}G(x,y)+ \frac{\partial}{\partial y}G(x,y)\right).
\end{equation}

The Lagrangian plays a key role in the large deviations of the Markov processes $\gamma_N$
\begin{equation}\label{2o11}
\mathcal L(f_1,f_2)(t)=\sup_{\varkappa_1(t),\varkappa_2(t)}\left\{\dot f_1(t)\varkappa_1(t)+\dot f_2(t)\varkappa_2(t)-
H(f_1(t),f_2(t),\varkappa_1(t),\varkappa_2(t))
\right\}, 
\end{equation}
which is the Legendre transform of the corresponding Hamiltonian
\begin{equation}\label{4o1}
H(f_1,f_2,\varkappa_1,\varkappa_2)=\lambda(1-f_1)[e^{\varkappa_1}-1]+\mu f_1[e^{-\varkappa_1-\varkappa_2}-1].
\end{equation}
Thus
\begin{equation}
 \begin{array}{rcl}
\dot f_1&=&\lambda (1-f_1)\exp\{\varkappa_1\}-\mu f_1\exp\{-(\varkappa_1+\varkappa_2)\}, \\
\dot f_2&=&-\mu f_1\exp\{-(\varkappa_1+\varkappa_2)\}.
\end{array}
\end{equation}
Here $f_1$ belongs to the set $C_1[0,T]$ of differentiable functions on $[0,T]$ with
values from interval $[0,1]$, and $f_2\in C_-[0,T]$ the set of differentiable 
functions on $[0,T]$ taking values in $\mathbb R_-$. 
The functions $\varkappa_1,\varkappa_2 \in C[0,T]$, where $C[0,T]$ is
the space of differentiable functions on $[0,T]$ taking values in
$\mathbb R$. The functions $\varkappa_1,\varkappa_2$ are called the {\it 
momenta}  of the dynamics, and the pair $(f_1,f_2)$ describes the
dynamics of the process $\gamma_N(t)$ at large $N$: the probability that
a trajectory of the process $\gamma_N$ is close to the path
$(f_1(t),f_2(t))$ is asymptotically represented by Lagrangian \reff{2o11}
$$
\mathbb{P} \bigl( (\gamma_N(t))_{t\in [0,T]} \approx (f_1(t),f_2(t))_{t\in [0,T]}\bigr) \approx \exp \Bigl( -N \int_0^T \mathcal{L}(f_1,f_2)(t)dt \Bigr).
$$
To derive \reff{4o1}, we follow the scheme in \cite{FeK}. First,
we have to construct a non-linear Hamiltonian
\begin{eqnarray*}
({ \cal H}_N G)(x_N,y_N)&:=&\frac1N \exp\{-NG(x_N,y_N)\}\times {\mathbf L}_{\gamma_N}\exp\{NG\}(x_N,y_N)\\
&=&\la(1-x_N)\left[\exp\left\{N\left(G\left(x_N+\tfrac{ 1}N,y_N\right)-G\left(x_N,y_N\right)\right)\right\}-1\right]\\
&&{} \ + \mu x_N\left[\exp\left\{N\left(G\left(x_N-\tfrac1N,y_N\right)-G\left(x_N,y_N\right)\right)\right\}-1\right].
\end{eqnarray*}
And, second, to find its limit as $N\to\infty$:
\begin{eqnarray}\label{2.13}
&& \lim_N({ \cal H}_N G)(x_N,y_N) \nonumber \\
&& {} =\la(1-x)\left[\exp\{\tfrac{\partial}{\partial x}G(x,y)\}-1\right]+\mu x\left[\exp\{-\tfrac{\partial}{\partial x}G(x,y)-\tfrac{\partial}{\partial y}G(x,y)\}-1\right].
\end{eqnarray}
The right-hand side of \reff{2.13} is the Hamiltonian \reff{4o1}
after the following changes of variables:
$$
f_1=x, f_2=y, \k_1=\tfrac{\partial}{\partial x}G(x,y), \k_2=\tfrac{\partial}{\partial y}G(x,y).
$$
We consider the problem of finding the optimal paths of large
radiations. Namely, we consider the following boundary conditions:
\begin{equation}\label{4o2}
\begin{array}{l}
f_1(0)=d \in[0,1],\\
f_2(0)=0,\ f_2(T)=-B, \mbox{ where }B>0,
\end{array}
\end{equation}
for the Hamiltonian system corresponding to $H$ in \reff{4o1}
\begin{equation}\label{ham1}
\left\{ \begin{array}{rcl}
\dot f_1&=&\lambda (1-f_1)\exp\{\varkappa_1\}-\mu f_1\exp\{-(\varkappa_1+\varkappa_2)\}, \\
\dot f_2&=&-\mu f_1\exp\{-(\varkappa_1+\varkappa_2)\}, \\
\dot \varkappa_1&=&\lambda \exp\{\varkappa_1\}-\mu\exp\{-(\varkappa_1+\varkappa_2)\} - \lambda + \mu, \\
\dot\varkappa_2&=&0.
\end{array} \right.
\end{equation}
We study the system for large $B$.

The solution of \reff{ham1} with the boundary conditions
\reff{4o2} presents the mean dynamics $f_1$ and $f_2$ of the occupation and
radiation processes. We add the boundary condition $\varkappa_1(T)=0$ because
the value of $f_1(T)$ is free.

\begin{remark} (On the large deviations principle.)
The book \cite{FeK} proposes a program how to prove the large deviation
principle for stochastic processes (in particular, Markov
processes) in consideration. This program requires special
considerations, depending on the studied process. For example, see
\cite{Kraa}, where the large deviation principle is studied for  a
model that is close to our model in some respect.

We, however, consider the Markov process  $\Gamma_N(t)$ which is a
scaled sum of $N$ independent Markov processes. Each one of these
Markov processes is concentrated on piece-wise constant paths on $[0,T]$.  We
can consider the Skorohod space of piece-wise continuous paths on
$[0,T]$ (see \cite{Bil}, Section 12). For our $N$ independent paths we can use the Sanov
theorem with subsequent application of the contraction principle (see
\cite{reza}). This combination of the Sanov theorem and the contraction
principle gives the rate function (see \reff{rate}) for our process $\Gamma$.
\end{remark}

\subsection{Results}

Assume  that the total radiation during the time interval $[0,T] $ is
$\lfloor BN \rfloor$, where $\lfloor \cdot\rfloor$ is integer part, i.e.
$\Upsilon_N(T)=-\lfloor BN\rfloor,\ B>0 $, with initial value
$\Upsilon_N(0)=0$. The initial number of excited atoms is $\Xi_N(0)=\lfloor
dN\rfloor \in[0,N]$. For example, we can choose $d=\frac{\la}{\la+\mu}$. Then
the  large deviation principle  yields
\begin{eqnarray}\label{ldp1}
&&\lim_{N\to\infty}\frac1N\ln\Pr\left(\xi(0)=\lfloor Nd\rfloor ,\ \zeta(0)=0,\zeta(T)=-\lfloor NB\rfloor\right)\\
&&\ \ \ \ \ \ \ \ \ \ \ \ = -\inf_{f_1,f_2}\left\{I(f_1,f_2):\:f_1(0)=d,\  f_2(0)=0,f_2(T)=-B\right\},\nonumber
\end{eqnarray}
where the rate function is, see \reff{2o11}, \reff{4o1},
\begin{eqnarray*}
I(f_1,f_2) &=& \int_0^T \mathcal{L} (f_1, f_2) \ed t \\
 &=& \int_0^T \sup_{\varkappa_1,\varkappa_2}\left(\varkappa_1\dot f_1+\varkappa_2\dot f_2-\lambda(1-f_1)[e^{\varkappa_1}-1]-\mu f_1[e^{-\varkappa_1-\varkappa_2}-1] \right)\ed t.
\end{eqnarray*}
\begin{remark}
The rate function is convex therefore the infimum in \reff{ldp1} is achieved and unique. Indeed, as the functional $I(f_1,f_2)$  is the convex combination of suprema of affine functionals it is convex. So its critical point is its minimum. This minimum is unique because the corresponding boundary value problem for Euler system has a unique solution as follows from calculations. 
\end{remark}

We rewrite the rate function in an alternative way.
\begin{equation}\label{rate}
I(f_1,f_2)=\int_0^T \sup_{\varkappa_1(t),\varkappa_2(t)}\left(\varkappa_1(t)\dot f_1(t)+\varkappa_2(t)\dot f_2(t)-\rho(f_1)[\varphi(\varkappa_1(t),\varkappa_2(t))-1]\right)\ed t,
\end{equation}
where $\rho(f_1)=\lambda(1- f_1)+\mu f_1$, and
\begin{equation}\label{laplace}
\varphi(\varkappa_1,\varkappa_2)=\frac{\lambda (1-f_1)}{\rho(f_1)}\exp\{\varkappa_1\}+\frac{\mu f_1}{\rho(f_1)}\exp\{-(\varkappa_1+\varkappa_2)\}.
\end{equation}
The supremum over $\varkappa_1(t),\varkappa_2(t)$ in \reff{rate} can be found
from the equations
\begin{eqnarray}
\dot f_1&=&\lambda (1-f_1)\exp\{\varkappa_1\}-\mu f_1\exp\{-(\varkappa_1+\varkappa_2)\},\nonumber\\
\dot f_2&=&-\mu f_1\exp\{-(\varkappa_1+\varkappa_2)\}.
\end{eqnarray}
The infimum over $f_1,f_2$ can be found from Euler equations
\begin{eqnarray}\label{euler}
\dot \varkappa_1&=&\lambda \exp\{\varkappa_1\}-\mu\exp\{-(\varkappa_1+\varkappa_2)\}-\lambda+\mu,\\
\dot\varkappa_2&=&0.\nonumber
\end{eqnarray}
In other words, we have to solve the Hamiltonian system of equations \reff{ham1} 
\begin{equation*}
\left\{ \begin{array}{rcl}
\dot f_1&=&\lambda (1-f_1)\exp\{\varkappa_1\}-\mu f_1\exp\{-(\varkappa_1+\varkappa_2)\}, \\
\dot f_2&=&-\mu f_1\exp\{-(\varkappa_1+\varkappa_2)\}, \\
\dot \varkappa_1&=&\lambda \exp\{\varkappa_1\}-\mu\exp\{-(\varkappa_1+\varkappa_2)\} - \lambda + \mu, \\
\dot\varkappa_2&=&0.
\end{array} \right.
\end{equation*}
under boundary conditions $f_1(0)=d$, $f_2(0)=0$, $f_2(T)=-B$ and
$\varkappa_1(T)=0$ with the Hamiltonian \reff{4o1}. We have derived the
system of equations \reff{ham1} with boundary conditions \reff{4o2} as a solution
of optimization problem. The solution of the boundary problem is unique and
gives an extremal path for the radiation problem.

\subsubsection{Solutions of \reff{ham1}} Two last equations of \reff{ham1} do not depend on $f_i, \dot f_i,\  i=1,2$. These equations can be solved as follows. The function $\varkappa_2$ does not depend on $t$. Having $\varkappa_2$ fixed we look for $\varkappa_1$ by means of the change of variable $y(t)=e^{\varkappa_1(t) }$. Then the third equation of \reff{ham1} is
\begin{equation}\label{yy}
\dot y=\lambda y^2-(\lambda-\mu)y-\mu e^{-\varkappa_2}.
\end{equation}
Using notations
\begin{equation}
\label{notation}
\begin{array}{l}
\gamma=\frac{\mu}{\lambda},\\[0.3cm]
a=1-\gamma,\\[0.3cm]
b=\gamma e^{-\varkappa_2}
\end{array}
\end{equation}
we reduce \reff{yy} to the form
\begin{equation}\label{yy1}
\dot y=\lambda\big(y^2-ay-b\big).
\end{equation}
The value of $a$ is upper bounded, $a\leq 1$, since $\gamma\geq 0$,
however $a$ may be negative and may have a large absolute value.
We represent the expression in the parentheses of \eqref{yy1} as
\begin{equation}\label{yy2}
y^2-ay-b=(y-r_1)(y-r_2),
\end{equation}
where
\begin{equation}\label{r12}
r_{1,2} = \frac{a}{2} \mp \sqrt{ \left( \frac{a}{2} \right)^2 + b }.
\end{equation}
Since $b>0$, the roots $r_1$ and $r_2$ have different signs:
$r_1<0<r_2$. The equation  \reff{yy1} can be represented as
\begin{equation}\label{yy22}
\frac{\dot y}{y-r_2}-\frac{\dot y}{y-r_1}=\lambda(r_2-r_1).
\end{equation}
The solution of \reff{yy2} is
\begin{equation}\label{sol2}
y(t)=e^{\varkappa_1 (t) } = \frac{r_2-r_1C_1\exp\{t\lambda(r_2-r_1)\}}{1-C_1\exp\{t\lambda(r_2-r_1)\}},
\end{equation}
where $ C_1 $ is a constant within the interval
$\big(\frac{r_2}{r_1},1\big)$. The value of $C_1$ cannot be out of this
interval because $0< y(t)< \infty$. Remark that $\frac{r_2}{r_1}<0$.

Note that $\k_1(T)=0$ since there are no constraints on the value
of $f_1(T)$. Then
\begin{equation}\label{C1}
C_1=\frac{r_2-1}{r_1-1}\exp\{-T\lambda(r_2-r_1)\}
\end{equation}
and
\begin{equation}\label{sol3}
e^{\varkappa_1(t)}=\frac{r_2+r_1\frac{r_2-1}{1-r_1}\exp\{(t-T)\lambda(r_2-r_1)\}}{1+\frac{r_2-1}{1-r_1}\exp\{(t-T)\lambda(r_2-r_1)\}}.
\end{equation}

The energy conservation law implies (see \reff{4o1}) that \textit{there
exists a constant $K$ such that, for any $t\in T$,}
$$
H(f_1,f_2,\varkappa_1,\varkappa_2)=K.
$$
Thus the path $f_1$ corresponding to the density of excited atoms is
\begin{equation}\label{f1}
f_1(t)=\frac{\lambda(e^{\varkappa_1}-1)-K}{\lambda(e^{\varkappa_1}-1)+\mu(1-e^{-\varkappa_1-\varkappa_2})}.
\end{equation}
The constant $K$ is sought from the initial condition $f_1(0)=d$. That
is\pch{,}
\begin{equation}\label{K}
K=\lambda(e^{\varkappa_1(0)}-1)(1-f_1(0))-f_1(0)\mu(1-e^{-\varkappa_1(0)-\varkappa_2}).
\end{equation}
Next we have to evaluate the emission path (see \reff{ham1})
\begin{equation}
f_2(t)=-\mu e^{-\varkappa_2}\int_0^tf_1(s)e^{-\varkappa_1(s)}\ed s,\ \ t\in[0,T].
\end{equation}
To this end we evaluate the integral
\begin{equation}\label{indef}
 J(t)= \int_0^t \frac{ \lambda(e^{\k_1(s)} -1) + K }{ \lambda(e^{\k_1(s)} -1) + \mu(1-e^{-\k_1(s)} e^{-\k_2})  } e^{-\k_1(s)} \ed s.
\end{equation}
The denominator of this integral multiplied by  $e^{\k_1(s)}$ can be
written as
$$
\lambda e^{2\k_1(s)} - (\lambda - \mu) e^{\k_1(s)} -\mu e^{-\k_2}= \lambda (e^{\k_1(s)} - r_1)(e^{\k_1(s)} - r_2),
$$
see \reff{yy1},\reff{yy2} and \reff{r12}. Hence
\begin{eqnarray*}
J(t)&=&  \int_0^t \frac{ \ed s }{ e^{\k_1(s)} -r_2} + \int_0^t \frac{ \lambda(r_1 -1) + K }{ \lambda (e^{\k_1(s)} - r_1)(e^{\k_1(s)} - r_2) } \ed s  \\
&=&\int_0^t \frac{\ed s }{ e^{\k_1(s)} -r_2}  +
\frac{ \lambda(r_1 -1) + K }{ \lambda (r_2- r_1) } \Bigl(  \int_0^t \frac{
\ed s }{ e^{\k_1(s)} -r_2} - \int_0^t \frac{\ed s }{ e^{\k_1(s)} -r_1} \Bigr) \\
&=&  \Bigl( 1+  \frac{ \lambda(r_1 -1) + K }{ \lambda (r_2- r_1) } \Bigr) \int_0^t \frac{ \ed s }{ e^{\k_1(s)} -r_2}  - \frac{ \lambda(r_1 -1) + K }{ \lambda (r_2- r_1) }  \int_0^t \frac{ \ed s }{ e^{\k_1(s)} -r_1}.
\end{eqnarray*}
\vspace{.3cm}

\noindent Since
\begin{equation*}
e^{\k_1(s)} -r_1 =\frac{ r_2 + r_1 \frac{r_2-1}{1-r_1} e^{\lambda(r_2-r_1)(s-T)} }{ 1+ \frac{r_2-1}{1-r_1} e^{\lambda(r_2-r_1)(s-T)}}
- r_1 = \frac{ r_2 - r_1}{ 1+ \frac{r_2-1}{1-r_1} e^{\lambda(r_2-r_1)(s-T)}}
\end{equation*}
then
\begin{eqnarray*}
J_1(t)&=&\int_0^t  \frac{\ed s }{ e^{\k_1(s)} -r_1} = \int_0^t \frac{ 1+ \frac{r_2-1}{1-r_1} e^{\lambda(r_2-r_1)(s-T)}}{ r_2 - r_1} \ed s \\
&=&  \frac{t}{r_2-r_1} + \frac{r_2-1}{1-r_1} \cdot \frac{e^{-\lambda(r_2-r_1)T}}{\lambda(r_2-r_1)^2} \Bigl( e^{\lambda(r_2-r_1)t} -1 \Bigr).
\end{eqnarray*}
\vspace{.3cm}

\noindent Since
\begin{equation*}
e^{\k_1(s)} -r_2 =\frac{ r_2 + r_1 \frac{r_2-1}{1-r_1} e^{\lambda(r_2-r_1)(s-T)} }{ 1+ \frac{r_2-1}{1-r_1} e^{\lambda(r_2-r_1)(s-T)}}
- r_2 = \frac{ (r_1 - r_2)\frac{r_2-1}{1-r_1}e^{\lambda(r_2-r_1)(s-T)}   }{ 1+ \frac{r_2-1}{1-r_1} e^{\lambda(r_2-r_1)(s-T)}}
\end{equation*}
then
\begin{eqnarray*}
J_2(t)&=&\int_0^t  \frac{ \ed s }{ e^{\k_1(s)} -r_2} = \int_0^t \frac{ 1+ \frac{r_2-1}{1-r_1} e^{\lambda(r_2-r_1)(s-T)}}{ (r_1 - r_2)\frac{r_2-1}{1-r_1}e^{\lambda(r_2-r_1)(s-T)} } \ed s \\
&=& - \frac{t}{r_2-r_1} -\frac{1-r_1}{r_2-1} \cdot \frac{e^{\lambda(r_2-r_1)T}}{\lambda(r_2-r_1)^2} \Bigl( 1 - e^{-\lambda(r_2-r_1)t}  \Bigr).
\end{eqnarray*}
Thus
\begin{equation}\label{I}
J(t)= \Bigl( 1+  \frac{ \lambda(r_1 -1) + K }{ \lambda (r_2- r_1) } \Bigr)J_1(t)- \frac{ \lambda(r_1 -1) + K }{ \lambda (r_2- r_1) }J_2(t)
\end{equation}
and
\begin{equation}\label{f2}
f_2(t)=-\mu\exp\{-\varkappa_2\} I(t).
\end{equation}
Therefore $\k_2$ is a solution of
\begin{equation}\label{kap22}
-B=f_2(T)=-\mu\exp\{-\varkappa_2\}\left[\Bigl( 1+  \frac{ \lambda(r_1 -1) + K }{ \lambda (r_2- r_1) } \Bigr)J_1(T)- \frac{ \lambda(r_1 -1) + K }{ \lambda (r_2- r_1) }J_2(T)\right].
\end{equation}
Remark that $e^{-\k_2}$ is present in the definition of constants
$r_1,r_2$ and $K$. Therefore \reff{kap22} is not a linear equation with
respect to $e^{-\k_2}$.

\section{Emergence of chaos in large fluctuations}

\subsection{Equations}
Here we derive a surprising feature of the asymptotic ``quasi-equilibrium"
behaviour of the emission and excitation rates within the interval $[0,T]$. We
will prove that these rates become equal in the limit and do not depend on
$\lambda$ and $\mu$.

For the hamiltonian system \reff{ham1} with boundary conditions \reff{4o2}, let
us make a change of variables $y(t)=e^{\varkappa_{1}(t)}$ and
$z(t)=e^{-\varkappa_{1}(t)-\varkappa_{2}(t)}$. Recall that
$\varkappa_{2}(t)=\varkappa_{2}$ is a constant. From (\ref{sol2}), we get
\begin{equation}\label{ey}
y(t)=\frac{r_{2}(1-r_{1})+r_{1}(r_{2}-1)\phi(t)}{1-r_{1}+(r_{2}-1)\phi(t)},
\end{equation}
where
\begin{equation}\label{ephi}
\phi(t)=e^{(t-T)\lambda(r_{2}-r_{1})}.
\end{equation}
We also have
\[
r_{2}-r_{1}=\frac{\sqrt{(\lambda-\mu)^{2}+4\mu\lambda e^{-\varkappa_{2}}}}{\lambda},
\]
(hence $0<\phi(t)\le 1$ for $0\le t\le T$) and
\begin{equation}\label{er2}
r_{2}=\frac{\lambda-\mu+\sqrt{(\lambda-\mu)^{2}+4\mu\lambda e^{-\varkappa_{2}}}}{2\lambda}.
\end{equation}
Note that the values of $r_{1},r_{2}$ and $\varkappa_{2}$ depend on the
parameter $B>0$ and the values of $\lambda>0$, $\mu>0$, and $T>0$ are fixed.

\subsection{Some bounds}

First, note that the right-hand side of (\ref{ey}) is monotone decreasing as
$\phi$ changes from $0$ to $1$ if $r_{2}>1$ and monotone increasing in $\phi$
at the same interval if $r_{2}<1$. Indeed, we have
\begin{eqnarray*}
y(t)&=&\frac{r_{1}(1-r_{1})+(r_{2}-r_{1})(1-r_{1})+r_{1}(r_{2}-1)\phi(t)}{1-r_{1}+(r_{2}-1)\phi(t)} \\
&=& r_{1}+\frac{(r_{2}-r_{1})(1-r_{1})}{1-r_{1}+(r_{2}-1)\phi(t)},
\end{eqnarray*}
where $(r_{2}-r_{1})(1-r_{1})>0$ and $1-r_{1}>0$. Hence,
\[
y(t)=r_{1}+\frac{c_{1}}{c_{2}+c_{3}\phi(t)},
\]
where $c_{1},c_{2}>0$ and where $c_{3}>0$ if $r_{2}>1$ and $c_{3}<0$ if
$r_{2}<1$. Moreover, the denominator $1-r_{1}+(r_{2}-1)\phi(t)$ is positive
for all $\phi(t)\in[0,1]$. Then the monotonicity follows.

We have $y(t)=r_{2}$ if $\phi(t)=0$ and $y(t)=1$ if $\phi(t)=1$. Since $0\le
\phi(t)\le 1$, we conclude that
\begin{equation}\label{ele}
y(t)\le\max\{1,r_{2}\}\le r_{2}+1,\quad 0\le t\le T.
\end{equation}

Integrating the second equation in (\ref{ham1}), we get
\begin{equation}\label{ez}
\mu\int_{0}^{T}z(t)f_{1}(t)dt=B.
\end{equation}
Since $0\le f_{1}(t)\le 1$, we conclude from (\ref{ez}) that
\begin{equation}\label{nez}
\int_{0}^{T}z(t)dt\ge\frac{B}{\mu}.
\end{equation}

From the third equation in (\ref{ham1}) and from the boundary condition
$\varkappa_{1}(T)=0$ we get the equality
\[
\varkappa_{1}(0)=\mu\int_{0}^{T}z(t)dt-\lambda\int_{0}^{T}y(t)dt+(\lambda-\mu)T.
\]
Next, we use (\ref{ey}) and (\ref{nez}) and derive the bound
\begin{equation}\label{nek}
\varkappa_{1}(0)\ge B-\lambda(r_{2}+1)T +(\lambda-\mu)T.
\end{equation}
Since $\varkappa_{1}(0)=\ln y(0)$ and, hence, $\varkappa_{1}(0)\le
\ln(r_{2}+1)$, we conclude from (\ref{nek}) that
\begin{equation}\label{ner}
\ln(r_{2}+1)\ge B-\lambda(r_{2}+1)T +(\lambda-\mu)T.
\end{equation}

In order to derive a lower bound for $r_{2}$, we denote $r=\lambda
T(r_{2}+1)$ and $b=B+(\lambda-\mu)T+\ln(\lambda T)$ and reduce
(\ref{ner}) to
\begin{equation}\label{nerr}
\ln(r)+r\ge b.
\end{equation}
We assume that $b>1$ (this is true for $B$ large enough). Let us consider two
cases:
 $ r>b$ and $r\leq b$.
One verifies  that (\ref{nerr}) implies
\begin{equation}\label{nerr2}
r\ge b-\ln(b).
\end{equation}
This inequality in the first case follows from $\ln b>0$, and in the second case it follows from $\ln r\leq\ln b$.
Hence, we have
\begin{equation}\label{ner2}
r_{2}\ge cB-d,
\end{equation}
for some $c,d>0$ and all sufficiently large  $B>0$.

A lower bound of the same form holds for $-r_{1}$ since
$r_{1}+r_{2}=1-\mu/\lambda$. Now, let us turn to (\ref{ephi}). Fix
some $\alpha>0$, $\alpha<T$ and note that
\begin{equation}\label{nephi}
\phi(t)\le e^{\alpha\lambda(r_{2}-r_{1})}
\end{equation}
for all $t\in[0,T-\alpha]$. It follows that both products
$r_{1}(r_{2}-1)\phi(t)$ and $(r_{2}-1)\phi(t)$ vanish uniformly for all
$t\in[0,T-\alpha]$ as $B\to\infty$ (the exponent dominates polynomials).
Therefore, from (\ref{ey}) we get
\begin{equation}\label{limy}
\lim_{B\to+\infty}\max_{0\le t\le T-\alpha}|y(t)-r_{2}|=0.
\end{equation}

Now, we find the asymptotics of $z(t)$ for $t\in[0,T-\alpha]$ as
$B\to+\infty$. To this end we express $e^{-\varkappa_{2}}$ via $r_{2}$ from
the equality (\ref{er2}):
\begin{equation}\label{ek2}
(2\lambda r_{2}-\lambda+\mu)^{2}=(\lambda-\mu)^{2}+4\mu\lambda e^{-\varkappa_{2}}.
\end{equation}
Namely, we get
\begin{equation}\label{ek3}
(2\lambda r_{2}+(\mu-\lambda))^{2}=(\mu-\lambda)^{2}+4\mu\lambda e^{-\varkappa_{2}}
\end{equation}
and then
\begin{equation*}
4\lambda^{2} r_{2}^{2}+4\lambda r_{2}(\mu-\lambda)+(\mu-\lambda)^{2}=(\lambda-\mu)^{2}+4\mu\lambda e^{-\varkappa_{2}},
\end{equation*}
that is,
\begin{equation}\label{ek4}
\lambda r_{2}^{2}+r_{2}(\mu-\lambda)=\mu e^{-\varkappa_{2}}
\end{equation}

Since $z(t)=\frac{e^{-\varkappa_{2}}}{y(t)}$, we conclude that
\begin{equation}\label{limz}
\lim_{B\to+\infty}\frac{z(t)}{r_{2}}=\frac{\lambda}{\mu}
\end{equation}
uniformly for all $t\in[0,T-\alpha]$.

\subsection{The turnpike theorem}

Let us turn to the first equation in (\ref{ham1}) and write it in the form
\begin{equation}\label{e1m}
\dot{f}_{1}(t)=\lambda y(t)(1-f_{1}(t))-\mu z(t)f_{1}(t),
\end{equation}
where, as follows from (\ref{limz}) and (\ref{limy}),
\begin{equation}\label{lims}
\lim_{B\to+\infty}\frac{\lambda y(t)}{\lambda r_{2}}=1\quad\text{and}\quad \lim_{B\to+\infty}\frac{\mu z(t)}{\lambda r_{2}}=1
\end{equation}
uniformly for all $t\in[0,T-\alpha]$.

We are ready to state and prove the main result of this section. There is a
striking resemblance with a series of turnpike theorems in mathematical
economics, see \cite{Mcz}.
\begin{theorem}
For any $\alpha>0$ and any $\e>0$, there exists $B_{0}>0$ such that
inequality
\[
|f_{1}(t)-\frac{1}{2}|<\e
\]
holds for all $t\in[\alpha,T-\alpha]$ whenever $B\ge B_{0}$.
\end{theorem}
\begin{proof}
We make a linear time change $\tau=\lambda r_{2}t$ and derive the
required assertion from the explicit solutions of resulting linear
differential equations.

Rewrite (\ref{e1m}) on $[0,T-\alpha]$ as
\begin{equation}\label{e2m}
\dot{f}_{1}(t)=\lambda r_{2}(1+\e_{1}(t))(1-f_{1}(t))-\lambda r_{2}(1+\e_{2}(t))f_{1}(t),
\end{equation}
where $|\e_{i}(t)|\le\e$ for $i=1,2$, and for all
$t\in[0,T-\alpha]$. Because of (\ref{lims}), this is true for all
$B$ large enough. After the time change, we have
\begin{equation}\label{e3m}
\frac{d}{d\tau}{f}_{1}(\tau)=(1+\e_{1}(\frac{\tau}{\lambda r_{2}}))(1-f_{1}(\tau))-(1+\e_{2}(\frac{\tau}{\lambda r_{2}}))f_{1}(\tau),
\end{equation}
or, equivalently,
\begin{equation}\label{e4m}
\frac{d}{d\tau}{f}_{1}(\tau)=1-2f_{1}(\tau)+\e_{0}(\tau),
\end{equation}
where $|\e_{0}(\tau)|\le\e$ for all $\tau\in[0,\lambda
r_{2}(T-\alpha)]$. It remains to note that $0\le f_{1}(0)\le 1$ for
all $B>0$ and, hence, $f_{1}(\tau)$ converges to $1/2$
exponentially and uniformly in $B>0$. Indeed, this follows from the
explicit solution
\[
x(\tau)=e^{-2\tau}x(0)+\int_{0}^{\tau}\e_{0}(s)e^{2(s-\tau)}ds,
\]
where we denoted $x=f_{1}-1/2$.
\end{proof}

\section{Microscopic approach. Quantum Hamiltonian Formalism.}\label{GenFun}

In this section we study the model described in the previous sections by the
so-called quantum Hamiltonian formalism \cite{Lloy96,Schu01}. This is
somewhat a misnomer as no quantum physics is used. The name comes from
certain formal analogies with linear operators used in quantum mechanics. In
essence, the fundamental idea is to write the generator as a matrix and then
to use linear and multi-linear algebra for explicit computations of
expectations.

\subsection{Notation and Definitions}

The first step is to consider not only the number of the excited atoms at
moment $t$ and the cumulative number of emissions during $[0,t]$, but the
state $\eta_i(t)$ of each atom $i$, $1\leq i \leq N$, at moment $t$. These
atoms form a particle system. A microstate, i.e., a microscopic configuration
of the particle
system, is denoted by $\eta = \{\eta_1,\dots,\eta_N\}$ where $\eta_i \in
\{0,1\}$ denotes the state of atom $i$. The value 0 corresponds to the ground
state, while 1 corresponds to the excited state. The state space of the
particle system is therefore $\S_N=\{0,1\}^N$. The number $\Upsilon_N(t)$ of
emissions up to time $t$ is a non-negative integer so that the state space $S_N$
of the process $(\eta(t),\Upsilon_N(t))$ is $S_N = \S_N \times \ints_0$. The
number of excited atoms at time $t$ is given by $\Xi_N(t) = \sum_{i=1}^N
\eta_i(t)$.

For a given configuration $\eta$ we shall use the shorthand $\eta^k$ for the
switched configuration with state variables \be \eta_i^k = \eta_i +
(1-2\eta_k) \delta_{k,i} \ee which corresponds to a switch of the state of
atom $k$ in the configuration $\eta$. The generator for the process then
reads \bea \label{Markovgenerator} \mathcal{L}_N f(\eta,\Upsilon_N) & = &
\sum_{k=1}^N \left[ \lambda (1-\eta_k) \left( f(\eta^k,\Upsilon_N) -
f(\eta,\Upsilon_N) \right) + \mu \eta_k\left( f(\eta^k,\Upsilon_N+1) -
f(\eta,\Upsilon_N) \right) \right]. \eea It is easy to see that one recovers
\eref{3o1} for functions $f(\Xi_N,\Upsilon_N)$ with $\Xi_N = \sum_{i=1}^N
\eta_i$.

In order to apply the quantum Hamiltonian formalism we map the states of a
single atom to the basis vectors $\mathfrak{e}_0 := (1,0)^T$ (corresponding
to the ground state) and $\mathfrak{e}_1 := (0,1)^T$ (corresponding to the
excited state), forming the canonical basis of $\C^2$. The superscript $T$
denotes transposition, i.e., these vectors are considered to be column
vectors. For the $N$-particle system a state $\eta$ is then mapped to the
tensor basis $\ket{\eta} = \mathfrak{e}_{\eta_1} \otimes
\mathfrak{e}_{\eta_2} \otimes \dots \otimes  \mathfrak{e}_{\eta_N}$ of
$(\C^2)^{\otimes N}$. Analogously we define row vectors
$\bra{\eta}=\ket{\eta}^T$ with orthogonality relation $\inprod{\eta}{\eta'} =
\delta_{\eta,\eta'}$ where $\inprod{\cdot}{\cdot}$ denotes the usual scalar
product. A Bernoulli product measure on $\S_N$ with density $\rho_0$ is thus
given by the vector \bel{Bernoulli} \ket{x} = ((1-\rho_0,\rho_0)^T){^{\otimes
N}} =  (1-\rho_0)^N ((1,x)^T)^{\otimes N} \ee where $x=\rho_0/(1-\rho_0)$ is
the fugacity. Normalization is reflected by the property $\inprod{s}{x} =1$
where $\bra{s} = \sum_{\eta\in\S_N} \bra{\eta} = (1,1,\dots,1)$ is the
so-called summation vector.

In order to consider also the number $\Upsilon_N$ of radiation events we
define infinite-dimensional vectors $\ket{\Upsilon_N}$ with components
$\ket{\Upsilon_N}_n =\delta_{n,\Upsilon_N}$ for $n\geq 0$. Then a state
$(\eta,\Upsilon_N)$ is mapped to the tensor product $\dket{\eta,\Upsilon_N}
:= \ket{\eta}\otimes\ket{\Upsilon_N}$ which is a vector in ${\mathbb V} :=
(\C^2)^{\otimes N} \otimes \C^{\ints_0}$ denoted by the double-ket symbol
$\dket{\cdot}$. In complete analogy to above we define also
$\dbra{\eta,\Upsilon_N} = \dket{\eta,\Upsilon_N}^T$ with the orthogonality
relation $\dinprod{\eta,\Upsilon_N}{\eta,\Upsilon_N'} =  \delta_{\eta,\eta'}
\delta_{\Upsilon_N,\Upsilon_N'}$. A probability measure $\mu$ on
$S_N=\{0,1\}^{N} \times \ints_0$ is thus represented by a vector \bel{measure}
\dket{\mu} = \sum_{(\eta,\Upsilon_N)\in S_N} \mu(\eta,\Upsilon_N)
\dket{\eta,\Upsilon_N}. \ee Normalization is reflected by the property
$\dinprod{s}{\mu} =1$ where $\dbra{s} = \sum_{\eta,\Upsilon_N}
\dbra{\eta,\Upsilon_N} =\bra{s}\otimes \bra{s'}$ with $\bra{s'} =
\sum_{\Upsilon_N\in\ints_0} \bra{\Upsilon_N}$.

Following the quantum Hamiltonian formalism we first introduce the
two-dimensional unit matrix  ${\mathbf 1}$ and the Pauli matrices
$\sigma^{x,y,z}$. In order to build operators for the $N$-particle system we
define the unit operator $\tilde{\mathbf 1} := {\mathbf 1}^{\otimes N}$ on
$(\C^2)^{\otimes N}$ and the unit operator  $\mathds{1}$ on $\C^{\ints_0}$.
In order to describe the evolution of the number of excited atoms we next
introduce the diagonal number operator $\hat{N} = \sum_{i=1}^N \hat{n}_i$
which is composed of the the sum of local operators $\hat{n}_i$ with the
property $\hat{n}_i \ket{\eta} = \eta_i \ket{\eta}$ and therefore $\hat{N}
\ket{\eta} = \Xi_N \ket{\eta}$. The subscript $i$ at a matrix $a_i$ denotes
the $i^{th}$ position in the tensor product $a_i = {\mathbf 1} \otimes \dots
\otimes {\mathbf 1} \otimes a \otimes {\mathbf 1} \dots \otimes {\mathbf 1}$.
In terms of Pauli matrices the single-atom number operator takes the form
$\hat{n} = 1/2({\mathbf 1}- \sigma^z)$. We also need the operator $\hat{V} =
\sum_{i=1}^N ({\mathbf 1}- \hat{n}_i)$ and the (non-diagonal) flipping
operators $\hat{S}^\pm = \sum_{i=1}^N \sigma_i^\pm$, where $\sigma^\pm =
1/2(\sigma^x \pm \rmi \sigma^y)$ with imaginary unit $\rmi$. The matrix
$\sigma_i^+$ turns excited atom $i$ into its ground state, while $\sigma_i^-$
corresponds to excitation of atom $i$.

For the radiation process we define the diagonal counting operator $\hat{K}$
where $\hat{K} \ket{\Upsilon_N} = \Upsilon_N \ket{\Upsilon_N}$. We also
introduce the matrix $\hat{K}^+$ where $\hat{K}^+$ acts  as raising operator
on  $\C^{\ints_0}$, i.e., $\hat{K}^+\ket{\Upsilon_N} = \ket{\Upsilon_N+1}$
with $\Upsilon_N \in \ints_0$. We note the commutation relation \bel{comm}
[\hat{K},\hat{K}^+]=\hat{K}^+. \ee which will play a role in the computation
of the radiation activity.

Finally, we lift the action of these operators to ${\mathbb V}$ by tensor
multiplication with the respective unit operators: $\hat{\mathbf N} = \hat{N}
\otimes \mathds{1}$, $\hat{\mathbf V} = \hat{V} \otimes \mathds{1}$,
$\hat{\mathbf S}^\pm = \hat{S}^\pm \otimes \mathds{1}$, $\hat{\mathbf K} =
\tilde{\mathbf 1} \otimes \hat{N}$, $\hat{\mathbf K}^+ = \tilde{\mathbf 1}
\otimes \hat{K}^+$. The generator of the process can now be written in matrix
form as \be {\mathbf H} = - \mu (\hat{\mathbf S}^+\hat{\mathbf
K}^+-\hat{\mathbf N}) - \lambda(\hat{\mathbf S}^- - \hat{\mathbf K}) =
\sum_{i=1}^N {\mathbf A}_i. \ee with \be {\mathbf A}_i = \mu (\hat{n}_i
\otimes \mathds{1} - \sigma_i^+\otimes \hat{K}^+) + \lambda(\tilde{\mathbf 1}
- \hat{n}_i - \sigma_i^- ) \otimes \mathds{1}. \ee For any initial measure
$\dket{\mu}$ the measure at time $t\geq 0$ is given by \be \dket{\mu(t)} =
\rme^{-{\mathbf H}t} \dket{\mu}. \ee Notice that the ${\mathbf A}_i$ commute
among themselves and therefore $\rme^{-{\mathbf H}t} = \prod_{i=1}^N
\rme^{-{\mathbf A}_it}$.

Consider now the initial measure $\dket{x}= \ket{x} \otimes \ket{0}$ with the
Bernoulli product measure $\ket{x}$ for the state of the atoms defined in
\eref{Bernoulli} and define
\bel{genfun}
F_T(x,y,z) := \E_x \left(y^{\Xi_N(T)} z^{\Upsilon_N(T)}\right)
\ee
which is the generating function
for the probability to arrive at $\Xi_N=M$ and $\Upsilon_N=K$ at time
$T$, starting at time 0 from $K=0$ (number of radiation events) and Bernoulli
product measure $\ket{x}$ for the state of the atoms.

Moreover define for $0\leq t \leq T$ the two-time expectations
\bea
G_T(x,y,z,t) & := & \E_x \left(\Xi_N(t) y^{\Xi_N(T)} z^{\Upsilon_N(T)}\right)
\\
H_T(x,y,z,t) & := & \E_x \left(\Upsilon_N(t) y^{\Xi_N(T)} z^{\Upsilon_N(T)}\right)
\eea
and the scaled and normalized quantities
\bea
g_T(x,y,z,t) & := & \frac{1}{N} G_T(x,y,z,t) / F_T(x,y,z) \\
h_T(x,y,z,t) & := & \frac{1}{N} H_T(x,y,z,t) / F_T(x,y,z).
\eea
They can be interpreted as conditional expectations of
the scaled variables $\xi_N(t)$ and $\zeta_N(t)$ of the scaled process
\eref{scaleproc}. Thus $g_T(x,y,z,t)$ is the conditional mean number of
excited atoms at time $t$ (normalized by the total number of atoms) and $h_T(x,y,z,t)$
is the conditional number of radiation events up to time $t$ per atom.
In the following we show that these functions are the microscopic analogs of
$f_1(t)$ and $f_2(t)$ of
the previous section. The case when $f_1(T)$ is free corresponds to $y=1$.


\subsection{Computation of the generating functions}

In the quantum Hamiltonian formalism we have by definition $F_T(x,y,z) =
\dbra{s} y^{\hat{\mathbf N}} z^{\hat{\mathbf K}} \rme^{-{\mathbf H}T}
\dket{x}$. Using the fact that $y^{\hat{N}}$ and  $z^{\hat{K}}$ are diagonal
and factorize we get that \bel{FT2} F_T(x,y,z) = \dbra{s} \rme^{-\hat{\mathbf
H}t} \dket{xy} \ee where $\hat{\mathbf H} = y^{\hat{\mathbf N}}
z^{\hat{\mathbf K}} \, {\mathbf H} \, y^{-\hat{\mathbf N}}  z^{-\hat{\mathbf
K}} =: \sum_{i=1}^N \hat{{\mathbf A}}_i$ with $\hat{{\mathbf A}}_i =
y^{\hat{n}_i}\otimes z^{\hat{K}} \, {\mathbf A}_i \, y^{-\hat{n}_i} \otimes
z^{-\hat{K}}$. Next we expand the exponential of $\hat{{\mathbf H}}$ into its
Taylor series and use the fact that $\bra{s} \hat{\mathbf K}^+ = \bra{s}$. It
follows that $\bra{s} \rme^{-\hat{{\mathbf H}}T} = \bra{s} \rme^{-\tilde{H}T}
\otimes \mathds{1} = \bra{s} \rme^{-\tilde{H}T} \otimes \bra{s'}$ with the
weighted generator $\tilde{H} = \sum_{i=1}^N \tilde{A}_i$ (acting on
$(\C^2)^{\otimes N}$) and \be \tilde{A}_i =  \mu (\hat{n}_i - y^{-1}z
\sigma_i^+) + \lambda({\mathbf 1} - \hat{n}_i - y \sigma_i^- ) . \ee

Observe that the action of $\tilde{H} \otimes \mathds{1}$ on the subspace
$\C^{\ints_0}$ is trivial. Hence the inner product in \eref{FT2} in this
subspace can be factorized and is trivially equal to one. Therefore
$F_T(x,y,z) = \bra{s} \rme^{-\tilde{H}T} \ket{xy}$ which is an inner product
only in $(\C^2)^{\otimes N}$. Moreover, due the factorization of both the
weighted generator and the initial measure we arrive at \bel{FT3} F_T(x,y,z)
= (1-\rho_0)^N \left[ (1,1) \rme^{-\tilde{A}T} (1,xy)^T \right]^N \ee with
the $2\times 2$ matrix \be \tilde{A} = \left( \ba{cc} \lambda & -\mu y^{-1}z
\\ - \lambda y & \mu \ea \right). \ee

A general $2\times 2$ matrix \be M = \left( \ba{cc} a &b \\ c & d \ea \right)
\ee has eigenvalues $\lambda_1 = (a+d - \delta)/2$ and $\lambda_2 = (a+d +
\delta)/2$ with $\delta = \sqrt{(a-d)^2 + 4bc}$. For $\delta \neq 0$ its
exponential takes the form \bel{expM} \rme^{-M\tau} =
\rme^{-\frac{(a+d)\tau}{2}} \left( \ba{cc}
\cosh{\frac{\delta}{2}\tau} + \frac{a-d}{\delta} \sinh{\frac{\delta}{2}\tau} & \frac{2 b}{\delta} \sinh{\frac{\delta}{2}\tau} \\
\frac{2 c}{\delta} \sinh{\frac{\delta}{2}\tau} & \cosh{\frac{\delta}{2}\tau}
- \frac{a-d}{\delta} \sinh{\frac{\delta}{2}\tau} \ea \right). \ee

For $\tilde{A}$ we have $\delta = \sqrt{(\mu-\lambda)^2 + 4\mu\lambda z}$ and
for convenience we define \bel{FT4} f_T(x,y,z) :=
\frac{\rme^{\frac{(\lambda+\mu)T}{2}} }{1-\rho_0}
\left(F_T(x,y,z)\right)^{\frac{1}{N}}. \ee This yields \bea f_T(x,y,z) & = &
(1,1) \left( \ba{cc}
\cosh{\frac{\delta}{2}T} + \frac{\mu-\lambda}{\delta} \sinh{\frac{\delta}{2}T} & \frac{2 y^{-1}z\mu}{\delta} \sinh{\frac{\delta}{2}T} \\
\frac{2 y\lambda}{\delta} \sinh{\frac{\delta}{2}T} & \cosh{\frac{\delta}{2}T}
- \frac{\mu-\lambda}{\delta} \sinh{\frac{\delta}{2}T} \ea \right) \left(
\ba{c} 1 \\ xy \ea \right) \nonumber
 \\
\label{FT5} & = &  (1+xy) \cosh{\frac{\delta}{2}T} + \frac{1}{\delta} \left[
2y\lambda + 2xz\mu + (\mu-\lambda)(1-xy)\right] \sinh{\frac{\delta}{2}T} .
\eea

Next we consider $G_T(x,y,z,t)$. According to definition one has \be
G_T(x,y,z,t) = \sum_{i=1}^N \dbra{s} y^{\hat{\mathbf N}} z^{\hat{\mathbf K}}
\rme^{-{\mathbf H}(T-t)} \hat{\mathbf N} \rme^{-{\mathbf H}t} \dket{x}. \ee
In a similar fashion to above one finds \be G_T(x,y,z,t) = N (1-\rho_0)^N
\rme^{\frac{-(\lambda+\mu)(N-1)T}{2}} f^{N-1}_T(x,y,z) \left[ (1,1)
\rme^{-\tilde{A}(T-t)} \hat{n} \rme^{-\tilde{A}t} (1,xy)^T \right], \ee and,
using \eref{expM}, \bea g_T(x,y,z,t) & = & \frac{1}{f_T(x,y,z)}
\left[ y \cosh{\frac{\delta}{2}(T-t)} +  \frac{2z \mu - y(\mu-\lambda)}{\delta} \sinh{\frac{\delta}{2}(T-t)} \right] \nonumber \\
\label{gT} & & \times \left[ x \cosh{\frac{\delta}{2}t} + \frac{2\lambda -
x(\mu-\lambda)}{\delta} \sinh{\frac{\delta}{2}t} \right] \eea with initial
condition \be g_T(x,y,z,0) = \frac{x}{f_T(x,y,z)} \left[ y
\cosh{\frac{\delta}{2}T} +  \frac{2z \mu - y(\mu-\lambda)}{\delta}
\sinh{\frac{\delta}{2}T} \right] \ee parametrized by the fugacity
$x=\rho_0/(1-\rho_0)$.

Next we note that \be H_T(x,y,z,t) = \dbra{s} y^{\hat{\mathbf N}}
z^{\hat{\mathbf K}} \rme^{-{\mathbf H}(T-t)} \hat{\mathbf K} \rme^{-{\mathbf
H}t}  \dket{x}. \ee In order to compute this quantity it is convenient to
consider its derivative w.r.t. $t$. Using $d/dt \, \rme^{{\mathbf H}t}
\hat{\mathbf K} \rme^{-{\mathbf H}t} = \rme^{{\mathbf H}t} [H,\hat{\mathbf
K}] \rme^{-{\mathbf H}t}$ and \eref{comm}, steps analogous to those that gave
rise to \eref{FT5} lead to \bea \dot{H}_T(x,y,z,t) & = & \mu y^{-1}z
\sum_{i=1}^N \dbra{s}  \rme^{-\hat{\mathbf H}(T-t)} \hat{\mathbf S}^+\hat{\mathbf K}^+  \rme^{-\hat{\mathbf H}t} \dket{xy} \nonumber \\
& = & N (1-\rho_0)^N \rme^{\frac{-(\lambda+\mu)(N-1)T}{2}} f^{N-1}_T(x,y,z)
\left[ (1,1) \rme^{-\tilde{A}(T-t)} \sigma^+ \rme^{-\tilde{A}t} (1,xy)^T
\right]. \eea Using \eref{expM} this yields \bea \dot{h}_T(x,y,z,t) & = &
\frac{z \mu }{f_T(x,y,z)}
\left[\cosh{\frac{\delta}{2}(T-t)} +  \frac{2y \lambda + \mu-\lambda}{\delta} \sinh{\frac{\delta}{2}(T-t)} \right] \nonumber \\
\label{hdotT} & & \times \left[ x \cosh{\frac{\delta}{2}t} +  \frac{2\lambda
- x(\mu-\lambda)}{\delta} \sinh{\frac{\delta}{2}t} \right]. \eea Integration
of this quantity w.r.t. $t$ is straightforward and yields $h(t)$ with initial
condition $h(0)=0$. The final values $g_T(x,y,z,T) \in [0,1]$ and
$h_T(x,y,z,T) \geq 0$ are parametrized by the variables $y,z$.

\begin{remark}
In the limit $T\to\infty$ straightforward computation for intermediate times
$t \propto cT$ shows that
\bea g^\ast(z) & := & \lim_{T\to\infty} g_T(x,y,z,cT) =
\frac{1}{2} \left(1 + \frac{\lambda-\mu}{\delta}\right) \\
\dot{h}^\ast(z)
& := & \lim_{T\to\infty} \dot{h}_T(x,y,z,cT) = \frac{z \lambda\mu}{\delta}
\eea
are constants that depend neither on the initial fugacity $x$ nor on the
atypical density of excited atoms parametrized by $y$. Moreover, for any choice of
$\lambda$ and $\mu$ the radiation activity $\dot{h}^\ast(z)$ is monotone in
$z$, i.e., one has $\dot{h}^\ast(z)\neq \dot{h}^\ast(1)$ for $z\neq 1$.
If $\lambda=\mu$ where radiation and excitation are in balance one finds
somewhat surprisingly that, despite the conditioning on a strongly atypical number
of excited atoms at time $T$ and a strongly atypical number of radiation events
up to time $T$, the mean number of excited atoms at intermediate times $t = O(T)$
is 1/2, which for $\lambda=\mu$ is the typical value without conditioning.
\end{remark}

\subsection{The case $\lambda=\mu=1$}

The expressions derived above simplify for $\lambda=\mu=1$ corresponding to
$\gamma=1$ in \eref{notation}. We present this case in more detail. One gets
$\delta = 2\sqrt{z}$ and it is convenient to introduce $\tilde{y} :=
y/\sqrt{z}$ and $\tilde{x} := x\sqrt{z}$. With this the expression obtained
above for the general case reduce to \bea \label{fgh1}
f_T(x,y,z) & = & \frac{1}{2} \left[ (1+\tilde{y})(1+\tilde{x}) \rme^{\sqrt{z}T} +(1-\tilde{y})(1-\tilde{x} ) \rme^{-\sqrt{z}T} \right] \\
\label{fgh2} g_T(x,y,z,t) & = & \frac{1}{2} \left( 1 - \frac{
(1+\tilde{y})(1-\tilde{x}) \rme^{\sqrt{z}(T-2t)} + (1-\tilde{y})(1+\tilde{x})
\rme^{-\sqrt{z}(T-2t)} }
{(1+\tilde{y})(1+\tilde{x}) \rme^{\sqrt{z}T} +(1-\tilde{y})(1-\tilde{x} ) \rme^{-\sqrt{z}T}} \right) \\
\label{fgh3} \dot{h}_T(x,y,z,t) & = & \frac{\sqrt{z}}{2} \left( h^\ast -
\frac{ (1+\tilde{y})(1-\tilde{x}) \rme^{\sqrt{z}(T-2t)} -
(1-\tilde{y})(1+\tilde{x}) \rme^{-\sqrt{z}(T-2t)} }
{(1+\tilde{y})(1+\tilde{x}) \rme^{\sqrt{z}T} +(1-\tilde{y})(1-\tilde{x} )
\rme^{-\sqrt{z}T}} \right) \eea with \be h^\ast =
\frac{(1+\tilde{y})(1+\tilde{x}) \rme^{\sqrt{z}T} - (1-\tilde{y})(1-\tilde{x}
) \rme^{-\sqrt{z}T}}{(1+\tilde{y})(1+\tilde{x}) \rme^{\sqrt{z}T}
+(1-\tilde{y})(1-\tilde{x} ) \rme^{-\sqrt{z}T}}. \ee

For large time frame $T$ one gets with the definition $\tau = T-t$ the
following asymptotic behaviour for $g(x,y,z,t) := \lim_{T\to\infty}
g_T(x,y,z,t) $:
\be
g(x,y,z,t) = \left\{ \ba{ll}
\frac{1}{2}\left(1- \frac{1-\tilde{x}}{1+\tilde{x}} \rme^{-2\sqrt{z}t} \right)& \mbox{ for } t = o(T) \\
\frac{1}{2} & \mbox{ for } t = O(T) \\
\frac{1}{2}\left(1- \frac{1-\tilde{y}}{1+\tilde{y}} \rme^{-2\sqrt{z}\tau}
\right)& \mbox{ for } \tau = o(T) \ea \right.
\ee

The corresponding
asymptotics for $h(x,y,z,t) := \lim_{T\to\infty} h_T(x,y,z,t) $ are: \be
h(x,y,z,t) = \left\{ \ba{ll}
\frac{\sqrt{z}}{2} t- \frac{1}{4} \frac{1-\tilde{x}}{1+\tilde{x}} \rme^{-2\sqrt{z}t} & t = o(T) \\
\frac{\sqrt{z}}{2} t & t = O(T) \\
\frac{\sqrt{z}}{2} t- \frac{1}{4} \frac{1-\tilde{y}}{1+\tilde{y}}
\rme^{-2\sqrt{z}\tau}  & \tau = o(T) \ea \right. \ee Thus, for $T$ large, the
radiation intensity $\dot{h}_T(x,y,z,t)$ is approximately constant except
initially and towards the end of the observation period $T$.

In order to make contact with the macroscopic results of the previous
sections we take $T$ fixed, choose $y=1$, and impose the initial condition
$g_T(x,y,z,t) =1/2$. This requires choosing the initial fugacity $x$ such that
$(1+\tilde{y})(1-\tilde{x}) \rme^{\sqrt{z}T} + (1-\tilde{y})(1+\tilde{x})
\rme^{-\sqrt{z}T} = 0$, i.e., for $\tilde{y} = 1/\sqrt{z}$ one chooses
\be \frac{1-\tilde{x}}{1+\tilde{x}} =
\frac{1-\sqrt{z}}{1+\sqrt{z}}  \rme^{-2\sqrt{z}T} =: C
\ee
Thus one obtains for $t\in[0,T]$
\bea
g_T(x,1,z,t) & = & \frac{1}{2} + \frac{C}{1-C^2}  \sinh{(2\sqrt{z}t)} \\
h_T(x,1,z,t) & = & \frac{\sqrt{z}}{2} \left(1-\frac{2}{1-C^2}\right) t +
\frac{C}{2(1-C^2)}  \sinh{(2\sqrt{z}t)}.
\eea
With the choice $z=\rme^{-\varkappa_2}$ (see \reff{notation}) one finds $f_1(t)=g_T(x,1,z,t)$
and $f_2(t)=h_T(x,1,z,t)$ (see \eref{f1} and \eref{f2} together
with \eref{K}, \reff{I} and \reff{kap22} for $\lambda=\mu=1$), thus relating
the macroscopic constant
$\varkappa_2$ appearing in the Hamiltonian \eref{4o1}
to the generating function parameter $z$ introduced in \eref{genfun}.

We point out two other special cases.
(a) For initial fugacity $x=1/\sqrt{z}$ and arbitrary $y$ one
has \bea
f_T(1/\sqrt{z},y,z,t) & = & (1+\tilde{y}) \rme^{\sqrt{z}T} \\
g_T(1/\sqrt{z},y,z,t) & = & \frac{1}{2}  \left( 1 - \frac{ 1-\tilde{y}}
{1+\tilde{y}} \rme^{-2\sqrt{z}(T-t)} \right) \\
h_T(1/\sqrt{z},y,z,t) & = & \frac{\sqrt{z}}{2} t + \frac{1}{4}  \frac{
1-\tilde{y}}{1+\tilde{y}} \left( \rme^{-2\sqrt{z}(T-t)} -
\rme^{-2\sqrt{z}T}\right). \eea Hence for initial times such that $T-t \gg
1/\sqrt{z}$ both the number of excited atoms and the radiation activity are
essentially constant. (b) On the other hand, for $y=\sqrt{z}$ and arbitrary initial
fugacity $x$ the expressions
\eref{fgh1} - \eref{fgh3} reduce to \bea
f_T(x,\sqrt{z},z,t) & = & (1+\tilde{x}) \rme^{\sqrt{z}T} \\
g_T(x,\sqrt{z},z,t) & = & \frac{1}{2}  \left( 1 - \frac{ 1-\tilde{x}}
{1+\tilde{x}} \rme^{-2\sqrt{z}t} \right) \\
h_T(x,\sqrt{z},z,t) & = & \frac{\sqrt{z}}{2} t + \frac{1}{4}  \frac{
1-\tilde{x}}{1+\tilde{x}}  \left(\rme^{-2\sqrt{z}t} - 1\right). \eea In this
case the number of excited atoms and the radiation activity approach a
constant  for times $t\gg 1/\sqrt{z}$.

\begin{remark} In all cases the radiation intensity is not constant, but has deviations from a constant in time windows
of length $1/(2\sqrt{z})$, i.e., the intensity approaches a constant (or
deviates from it) exponentially with a decay rate $1/(2\sqrt{z})$. We refer
to this behaviour as the occurence of ``moderate bursts''.
\end{remark}

\begin{remark} The quantum Hamiltonian approach makes clear that the behaviour of a two-level system is generic.
In a $m$-level system the functions $g$ and $h$ (and their analogs for
similar observables) will always be of the factorized form $g_T(\dots,t) =
a(\dots,T-t) b(\dots,t)$ and the functions $a$ and $b$ will be sums of
exponentials of the eigenvalues of the single-atom generator $A$.
\end{remark}

\section*{Acknowledgements}This work was supported by  FAPESP (S\~ao Paulo Research Foundation) grant 2015/03452-5.

A. Yambartsev   thanks CNPq (Conselho Nacional de Desenvolvimento
Cient\'ifico e Tecnol\^ogico) grant 307110/2013-3.

G.M. Sch\"utz thanks USP for Kind Hospitality

The researches of E. Pechersky and A. Vladimirov were carried out
at the Institute for Information Transmission Problems (IITP), Russian
Academy of Science and funded by the Russian Foundation for Sciences (grant
14-50-00150).

S. Pirogov studies have been funded by the Russian Science Foundation (project N 17-11- 01098).

\end{document}